\documentclass[twocolumn, prl]{revtex4}
\usepackage{amsmath,amssymb,amsthm,epsfig,graphics,graphicx}
\usepackage{verbatim, enumerate}
\usepackage{color}
\usepackage{dcolumn}
\usepackage{bm}

\newtheorem{theorem}{Theorem}

\theoremstyle{definition}

\newtheorem*{remark}{Remark}

\newcommand{\ket}[1]{|#1\rangle}

\begin{document}
\date{\today}
\title{\Large {\bf Local Transformations Requiring Infinite Rounds of Classical Communication}}
\author{Eric Chitambar}
\affiliation{Center for Quantum Information and Quantum Control (CQIQC),
Department of Physics and Department of Electrical \& Computer Engineering,
University of Toronto, Toronto, Ontario, M5S 3G4, Canada}

\begin{abstract}
In this paper, we study the number of rounds of communication needed to implement certain tasks by local quantum operations and classical communication (LOCC).  We find that the class of LOCC operations becomes strictly more powerful as more rounds of classical communication are permitted.  Specifically, for every $n$, there always exists an $n$ round protocol that is impossible to implement in $n-2$ rounds.  Furthermore, we show that certain entanglement transformations are possible if and only if the protocol uses an infinite (unbounded) number of rounds.  Interestingly, the number of rounds required to deterministically distill bipartite entanglement from a single multipartite state can be strongly discontinuous with respect to the amount of entanglement distilled.
\end{abstract}

\maketitle

Despite its importance to quantum communication \cite{Nielsen-2000a} and entanglement theory \cite{Plenio-2007a}, the class of Local Operations with Classical Communication (LOCC) is still not satisfactorily understood.  For instance, very little is known about what new operational possibilities become available using LOCC as more rounds of measurement and communication are performed.  If one allows the system dimensions to vary, Xin and Duan have constructed a collection of states in $m\otimes n$ systems that need at least $2\min\{m,n\}-2$ rounds of classical communication in order to be perfectly distinguished \cite{Xin-2008a}.  However, for a system of fixed dimensions, it is unclear how the power of LOCC depends on the number of turns the parties take in the measurement and communication process.  Bennett \textit{et al.} have shown that for the task of distilling EPR pairs from some mixed state, two-round communication between the parties is strictly more powerful than one-round \cite{Bennett-1996a}.  On the other hand, a well-known result of Lo and Popescue says that that every bipartite pure state transformation can be reduced to a single round of measurement and outcome broadcast \cite{Lo-1997a}.   

In this paper, we consider pure state transformations within systems of three qubits and find that the necessary number of rounds is vastly different than in the bipartite scenario.  The specific LOCC multi-outcome transformation we study is the following:
\begin{equation}
\label{Eq:trans}
\ket{W}\to\begin{cases}\ket{\Phi^{(AB)}}\;\;\text{with probability}\; p_{AB},\\
\ket{\Phi^{(AC)}}\;\;\text{with probability}\; p_{AC},\\
\ket{\varphi^{(BC)}}\;\;\text{with probability}\; p_{BC}
\end{cases}
\end{equation} 
where $\ket{W}=\sqrt{1/3}\left(\ket{100}+\ket{010}+\ket{001}\right)$, $\ket{\Phi^{ij}}$ is a maximally entangled two-qubit state shared between parties $i$ and $j$, and $\ket{\varphi^{(BC)}}$ is some entangled state held by Bob and Charlie.  This transformation is known as a random distillation since the two parties who end up sharing entanglement is unspecified before the LOCC process.  When $\ket{\varphi^{(BC)}}$ is also a maximally entangled state, transformation \eqref{Eq:trans} is the exact problem studied by Fortescue and Lo \cite{Fortescue-2007a}.  There, the authors showed that for any $\epsilon>0$, the above process can be accomplished with probability $p_{AB}+p_{BC}+p_{AC}>1-\epsilon$.  Here, we are concerned exclusively with deterministic transformations where $p_{AB}+p_{BC}+p_{AC}=1$ and ask how the amount of entanglement in $\ket{\varphi^{(BC)}}$ affects the number of LOCC rounds needed to accomplish the transformation.  

\begin{figure}[b]
\includegraphics[scale=0.5]{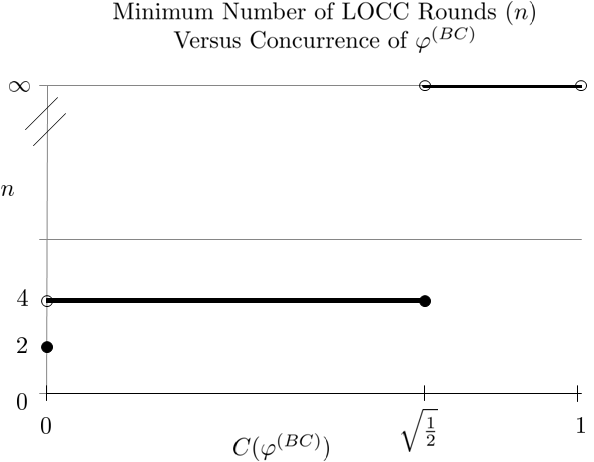}% Here is how to import EPS art
\caption{\label{roundplot}
The minimum number of rounds needed to perform transformation \eqref{Eq:trans} as a function of the entanglement in $\ket{\varphi^{(BC)}}$ (measured by the concurrence) when $p_{AB}+p_{AC}+p_{BC}=1$ and $p_{ij}>0$ for all pairs.  Here the point $\infty$ means that an unbounded number of rounds are required to achieve the transformation.} 
\end{figure}

We are able to prove that the minimum number of LOCC rounds required to perform this transformation is strongly discontinuous with respect to the entanglement of $\ket{\varphi^{(BC)}}$.  Furthermore, once the entanglement of $\ket{\varphi^{(BC)}}$ gets too large, an infinite number of LOCC rounds is necessary to complete the transformation with probability one (see Fig. \ref{roundplot}).  In addition, we are able to show that for any $n$, there exists transformations that require at least $n+2$ number of rounds, and thus LOCC becomes strictly stronger as more rounds of measurement and communication are permitted.  We turn now to a brief overview of the terminology and notation used in this paper.

An LOCC operation consists of each party taking turns to perform a ``local'' measurement on his/her subsystem and publicly announcing the result so to possibly affect the particular choice of future measurements.  A single measurement performed by one of the parties and the subsequent broadcast of that result constitutes one \textbf{round} in the LOCC operation.  We will say that any local unitary (LU) operation \textit{does not} consume one round of action.  Thus a bit more formally we can define an LOCC protocol as a fixed set of instructions that (i) identify a single party as the acting agent in each round, (ii) specify the particular measurement that party is to perform given the measurement outcomes in all previous rounds, and (iii) describe any LU operations to be performed by the other parties given the outcome of the measurement in (ii).  Included in these instructions is a halt command which indicates the end of the protocol whenever certain sequences of measurement outcomes are obtained.  A \textbf{finite round} LOCC protocol is one that necessarily halts after $n$ rounds for some $n\in\mathbb{Z}_+$; an \textbf{infinite round} LOCC protocol is one that does not.  When a protocol halts, its resultant state is known as a halt state, and we say that some LOCC protocol obtains state $\ket{\varphi_i}$ with probability $p_i$ if this is the total probability of $\ket{\varphi_i}$ being some halt state.

\begin{comment}
\begin{figure}[t]
\includegraphics[scale=0.6]{LOCCfig}% Here is how to import EPS art
\caption{\label{LOCCfig}
A tree diagram of an LOCC protocol lasting three rounds.  Alice and Bob alternate performing two-outcome measurements in each round thus generating eight different outcome branches.} 
\end{figure}

Every LOCC protocol can be represented as a tree in which each node represents a measurement and the nodes are aligned according to which round their measurement takes place (see Fig. \ref{LOCCfig}).  Each possible post-measurement state is represented by an edge which is either connected to some node in the next round, or it is a halt state and is identified as such.  We will say that one \textbf{branch} in an LOCC protocol refers to a specific path starting from the initial state, or first node of the tree, and ending at a halt state.  Each branch has a certain probability of being realized with the sum of all branch probabilities equaling unity.
\end{comment}

Our discussion focuses on converting the three qubit W state into two qubit entanglement by LOCC.  Any tripartite entangled state obtainable from $\ket{W}$ by LOCC with some nonzero probability is said to belong to the W-class of states \cite{Dur-2000a}, and up to an LU transformation, it will have a unique representation of the form $\sqrt{x_0}\ket{000}+\sqrt{x_1}\ket{100}+\sqrt{x_2}\ket{010}+\sqrt{x_3}\ket{001}$.  This uniqueness allows us to identify each W-class state with the three-component vector
\begin{align}
\vec{x}=(x_1&,x_2,x_3)\notag\\
&\updownarrow\notag\\
\sqrt{x_0}\ket{000}+\sqrt{x_1}\ket{100}+&\sqrt{x_2}\ket{010}+\sqrt{x_3}\ket{001},
\end{align} and $x_0=1-\sum_{i=1}^Nx_i$ \cite{Kintas-2010a}.  We will use this notation throughout the paper.

For two qubit pure states, a useful quantifier of entanglement is the concurrence measure \cite{Wootters-1998a}.  When parties $i$ and $j$ share the state $\ket{\varphi^{(ij)}}$, we denote its concurrence by $C(\varphi^{(ij)})$.  A maximally entangled state $\ket{\Phi^{(ij)}}$ satisfies $C(\Phi^{(ij)})=1$.  Because the concurrence is invariant under LU operations and we disregard such operations when counting the rounds of LOCC, the concurrence suffices in characterizing our halt states.  We can now state our first result concerning feasibility of transformation \eqref{Eq:trans}.

\begin{theorem}
\label{Thm1}
Let $C(\varphi^{(BC)})=t$ and consider transformation \eqref{Eq:trans}.  {\bf \upshape (I)} If $t>\sqrt{\frac{1}{2}}$, then there exists no finite round LOCC protocol such that $p_{AB}+p_{AC}+p_{BC}=1$.  Conversely, if $t\leq\sqrt{\frac{1}{2}}$, then there exists a protocol of three rounds satisfying $p_{AB}+p_{AC}+p_{BC}=1$; at least four are required when $p_{ij}>0$ for all pairs.  {\bf \upshape (II)} For any fixed $t\in(0,\sqrt{\frac{1}{2}}]$, if there is an $n$ round LOCC protocol $\mathcal{P}$ that achieves the transformation with $p_{AB}+p_{AC}+p_{BC}=1$, then there exists an $n+2$ round protocol $\mathcal{P}'$ that achieves the transformation with $p'_{AB}+p'_{AC}+p_{BC}'=1$ and $p'_{AB}+p'_{AC}>p_{AB}+p_{AC}$.
\end{theorem}
\begin{remark}
For part (II), the set of probabilities $p_{AB}+p_{AC}$ obtainable by $n$ round LOCC is, in fact, compact \cite{Chitambar-2011a}.  Consequently, there exists an optimal LOCC protocol $\mathcal{P}$ and therefore (II) implies that $\mathcal{P}'$ is not achievable in $n$ rounds.
\end{remark}

\begin{proof} 

(I)   

The proof makes heavy use of the result by Gour \textit{et. al} \cite{Gour-2005a} which says that a deterministic transformation $\ket{\psi}_{ijk}\to\ket{\varphi}_{ij}$ is possible iff 
\begin{equation}
\label{Eq:COA}
C^{(k)}_a(\psi^{(ijk)})\geq C(\varphi^{(ij)})
\end{equation}
where $C^{(k)}_a$ is the three-party concurrence of assistance with respect to party $k$ ``providing the assistance.''  Furthermore, when Eq. \eqref{Eq:COA} is satisfied, the transformation $\ket{\psi}_{ijk}\to\ket{\varphi}_{ij}$ requires no more than 2 rounds with only 1 round required when it is an equality \cite{Gour-2005a}.  

For a W-class state $\vec{x}=(x_1,x_2,x_3)$ with $x_0=0$, the concurrence of assistance is easily computed to be $C^{(k)}_a(\vec{x})=2\sqrt{x_ix_j}<1$, where the strict inequality follows from $x_k>0$.  For a general W-class state with $x_0\not=0$, the state $(x_1,x_2,x_3)$ can always be deterministically obtained from the state $(x_1,x_2,x_0+x_3)$ (which has no $\ket{000}$ component) \cite{Kintas-2010a}, and since $C^{(k)}_a$ is an entanglement monotone \cite{Gour-2005a}, it follows that 
\begin{equation}
\label{Eq:CaWstate}
C^{(k)}_a(\vec{x})<1
\end{equation} for any general three qubit W-class state $\vec{x}$.  

Now for some fixed value of $C(\varphi^{(BC)})=t$, consider a more general transformation in which the protocol halts after obtaining either $\ket{\Phi^{(AB)}}$, $\ket{\Phi^{(AC)}}$, or any $\ket{\tilde{\varphi}^{(BC)}}$ such that $C(\tilde{\varphi}^{(BC)})\geq C(\varphi^{(BC)})$.  It is well-known that this latter condition is both necessary and sufficient for the deterministic transformation of $\ket{\tilde{\varphi}^{(BC)}}\to\ket{\varphi^{(BC)}}$ \cite{Nielsen-1999a}.  Hence, any finite round transformation that deterministically obtains some state in the set $\{\ket{\Phi^{(AB)}}, \ket{\Phi^{(AC)}}, \ket{\tilde{\varphi}^{(BC)}}:C(\tilde{\varphi}^{(BC)})\geq C(\varphi^{(BC)})\}$ can be modified into a transformation that satisfies Theorem 1.

Suppose that this more general transformation can be accomplished in $n$ rounds and \textit{no fewer}.  Then at least one of the $n^{th}$ round pre-measurement states must be a W-class state $(x_1,x_2,x_3)$, and the $n^{th}$ round measurement on this state must convert it into either $\ket{\Phi^{(AB)}}$, $\ket{\Phi^{(AC)}}$, or states $\ket{\tilde{\varphi}^{(BC)}}$ with probability one.  From the concurrence conditions argued above (Eqns. \eqref{Eq:COA} and \eqref{Eq:CaWstate}), it follows that Alice must be the acting party in round $n$.  Consequently, there must exist some round $m<n$ such that it is the last round in which either $\ket{\Phi^{(AB)}}$ or $\ket{\Phi^{(AC)}}$ is an obtainable post-measurement halt state.  If we assume (without loss of generality) that it is the state $\ket{\Phi^{(AB)}}$, then at least one $m^{th}$ round pre-measurement state must be of the form $(\tfrac{1-s}{2},\tfrac{1-s}{2},s)$, and Charlie is the acting agent in the $m^{th}$ round.  One outcome of Charlie's measurement on this state is $\ket{\Phi^{(AB)}}$, and another will be $\vec{a}=(\tfrac{1-a_0-a_3}{2},\tfrac{1-a_0-a_3}{2},a_3)$ with $a_3>0$.  Again, because this state can be deterministically obtained from a state of the form $(\frac{1-b}{2},\frac{1-b}{2},b)$, the monotonicity of $C^{(A)}_a$ implies that $C^{(A)}_a(\vec{a})\leq 2\sqrt{\frac{(1-b)b}{2}}\leq \sqrt{\frac{1}{2}}$.  However, state $\vec{a}$ must be deterministically transformed into states $\ket{\psi_{BC}'}$ since no other types of states are obtainable after round $m$.  Hence by \eqref{Eq:COA}, it follows that $t\leq t'\leq \sqrt{\frac{1}{2}}$.

Now conversely, suppose that $t\leq\sqrt{\frac{1}{2}}$.  Consider a two-outcome measurement defined by the set $\{M_1(x),M_2(x)\}$ where 
\begin{align}
\label{Eq:Ops}
M_1(x)&=diag[\sqrt{x},0]&M_2(x)&=diag[\sqrt{1-x},0].
\end{align}
A very simple protocol that implements transformation \eqref{Eq:trans} is as follows.  In the first round, on the W state $(\tfrac{1}{3},\tfrac{1}{3},\tfrac{1}{3})$, Charlie performs $\{M_1(\frac{1}{2}),M_2(\frac{1}{2})\}$.  The first outcome is state $\ket{\Phi^{(AB)}}$ while the second is $(\frac{1}{4},\frac{1}{4},\frac{1}{2})$.  This latter state has $C^{(A)}_a=\sqrt{\frac{1}{2}}$, and so by Eq. \eqref{Eq:COA} $\ket{\varphi^{(BC)}}$ can be obtained with probability 1 in no more than 2 additional rounds.  One drawback to this protocol, however, is that $p_{AC}=0$.  A more sophisticated protocol avoids this problem.  In fact, it is easy to see that, based on the arguments we've given, the following protocol is optimal in the sense that at least four rounds are required if $p_{ij}>0$ for all pairs.  

\medskip
\noindent\textbf{- $\ket{W}$ Distillation Protocol (see Fig. \ref{Protocolfig}):}

\noindent\textbf{Round 1:}  Charlie performs the measurement $\{M_1(\alpha),M_2(\alpha)\}$ with $\alpha=\frac{2\sqrt{1-t^2}}{1+\sqrt{1-t^2}}$.  Outcome $1$ is the halt state $\ket{\Phi^{(AB)}}$, and it is reached with probability 
\begin{equation}
\label{Eq:PAB}
p_{AB}=\frac{2}{3}\alpha.
\end{equation} Outcome $2$ is $\ket{\phi_{2}}:=\frac{1}{3-2\alpha}(1-\alpha,1-\alpha,1)$.  \textbf{Round 2:}  On $\ket{\phi_{2}}$, Alice performs the same measurement $\{M_1(\alpha), M_2(\alpha)\}$.  Outcome $1$ is the halt state $\frac{1}{2-\alpha}(0,1-\alpha,1)$ obtained with total probability 
\begin{equation}
\label{Eq:PBC}
p_{BC}=\frac{2}{3}\alpha-\frac{1}{3}\alpha^2. 
\end{equation}
 This state has a concurrence of $t$.  For outcome $2$, the residual state is $\ket{\phi_{2,2}}:=\frac{1}{3-\alpha}(1,1-\alpha,1)$.  \textbf{Round 3:}  On $\ket{\phi_{2,2}}$, Bob performs the measurement $\{M_1(\beta),M_2(\beta)\}$ with $1-\beta=(1-\alpha)s$ where $s=\left(\frac{1+\sqrt{1-2t^2}}{2t}\right)^2$.  Outcome 1 outputs $\ket{\Phi^{(AC)}}$ with some nonzero probability, while the outcome 2 state is $\ket{\phi_{2,2,2}}:=\frac{1}{1+2s}(s,1,s)$.  It can readily be checked that $1+2s=\frac{2\sqrt{s}}{t}$ so that $\ket{\phi_{2,2,2}}$ has $C^{(A)}_a$ equaling $t$.  \textbf{Round 4:}
On $\ket{\phi_{2,2,2}}$ Alice is able to make a projective measurement that deterministically outputs $\ket{\varphi^{(BC)}}$ \cite{Gour-2005a}.

\begin{figure}[t]
\includegraphics[scale=0.5]{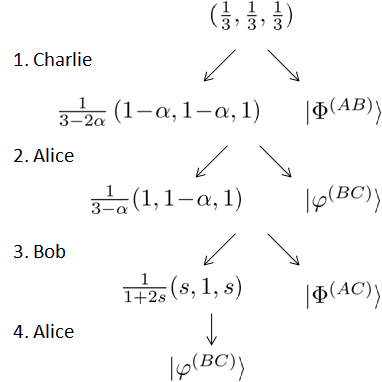}% Here is how to import EPS art
\caption{\label{Protocolfig}
A four round protocol that achieves transformation \eqref{Eq:trans} with total probability one.  Each party performs a two-outcome measurement and the resultant states are given.  Here, $(x_1,x_2,x_3)$ is notation for the state $\sqrt{x_1}\ket{100}+\sqrt{x_2}\ket{010}+\sqrt{x_3}\ket{001}$.  Also, $C(\varphi^{(BC)})=t$, $\alpha=\tfrac{2\sqrt{1-t^2}}{1+\sqrt{1-t^2}}$, and $s=(\tfrac{1+\sqrt{1-2t^2}}{2t})^2$.} 
\end{figure}

(II)  As in the proof of (I), for any $n$ round protocol $\mathcal{P}$ in which $p_{AB}+p_{AC}+p_{BC}=1$, there must exist some last round $m<n$ in which either $\ket{\Phi^{(AB)}}$ or $\ket{\Phi^{(AC)}}$ is a post-measurement halt state.  Consequently, after $m$ rounds the total probability of Alice sharing a maximally entangled state with either Bob or Charlie is $p_{AB}+p_{AC}$.   Now assuming again (without loss of generality) that $\ket{\Phi^{(AB)}}$ is an $m^{th}$ round halt state, in protocol $\mathcal{P}$ there must be some $m^{th}$ round pre-measurement state of the form $(\frac{1-s}{2},\frac{1-s}{2},s)$.   From this state, $\ket{\Phi^{(AB)}}$ is obtained with probability $q_{AB}$, and states $\vec{x}_i=(x_{1i},x_{2i},x_{3i})$ satisfying $C_a^{(A)}(\vec{x}_i)\geq t$ are each obtained with probability $q_i$.  It is not difficult to see that when Charlie acts, the quantities $x_2$ and $x_0+x_3$, always remain invariant on average \cite{Kintas-2010a}.  This implies $\frac{1-s}{2}=\frac{q_{AB}}{2}+\sum_iq_ix_{2i}$ and $s=\sum_iq_i(x_{0i}+x_{3i})$.
From these two equations and the fact that $2\sqrt{x_{2i}(x_{0i}+x_{3i})}\geq C_a^{(A)}(\vec{x}_i)\geq t$ for all $\vec{x}_i$, we have
\begin{align}
\label{Eq:CASW}
\sqrt{2(1-s-q_{AB})s}&=2\sqrt{\sum_iq_ix_{2i}}\sqrt{\sum_iq_i(x_{0i}+x_{3i})}\notag\\
&\geq \sum_iq_i2\sqrt{x_{2i}(x_{0i}+x_{3i})}\notag\\
&\geq(1-q_{AB})t
\end{align}
where we have used the Cauchy-Schwartz inequality.

Protocol $\mathcal{P}'$ follows $\mathcal{P}$ on all instructions except the $m^{th}$ round measurement on the state $(\frac{1-s}{2},\frac{1-s}{2},s)$.  Here, Charlie performs the measurement $\{M_1(\delta),M_2(\delta)\}$ with $\delta=\frac{q_{AB}}{1-s}$.  The first outcome is $\ket{\Phi^{(AB)}}$ with probability $q_{AB}$, and the other is $(\frac{1-s'}{2},\frac{1-s'}{2},s')$ where $s'=\frac{s}{1-q_{AB}}$.  The $C^{(A)}_a$ of this state is $\sqrt{2s'(1-s')}=\frac{1}{1-q_{AB}}\sqrt{2s(1-s-q_{AB})}\geq t$ by Eq. \eqref{Eq:CASW}.  On the state $(\frac{1-s'}{2},\frac{1-s'}{2},s')$, the following measurement scheme is implemented.  If $s'<\frac{1}{2}$, Charlie performs $\{M_1(\gamma), M_2(\gamma)\}$ with $\gamma=\frac{1-2s'}{(1-s')^2}$. Outcome 1 is $\ket{\Phi^{(AB)}}$ occurring with some nonzero probability while outcome 2 is a state with $C^{(A)}_a=\sqrt{2s'(1-s')}\geq t$.  From here, $\ket{\varphi^{(BC)}}$ can be obtained in no more than 2 rounds.  If $s\geq\frac{1}{2}$, then the state of the same form as $\ket{\psi_2}$ in the above protocol.  Rounds 2, 3, and 4 of that protocol are then executed on $(\sqrt{\frac{1-s'}{2}},\sqrt{\frac{1-s'}{2}},s')$ yielding $\ket{\Phi^{(AC)}}$ with a nonzero probability.  In total then, the new protocol $\mathcal{P}'$ runs no more than $n+2$ rounds with $p'_{AB}+p'_{AC}>p_{AB}+p_{AC}$.
\end{proof}
\begin{remark}
In part (II), it is crucial that Bob and Charlie are entangled for all other outcomes besides $\ket{\Phi^{(AB)}}$ and $\ket{\Phi^{(AC)}}$.  Indeed, consider the W state $(\frac{1}{3},\frac{1}{3},\frac{1}{3})$ where Charlie performs a projective measurement in the $\{\ket{0},\ket{1}\}$ basis.  In this case, $p_{AB}=\frac{2}{3}$, and there can be no other transformation with $p'_{AB}+p'_{AC}>\frac{2}{3}$ since the maximum probability of distilling 1 Ebit across the $A:BC$ cut is $\frac{2}{3}$ for the W state \cite{Lo-1997a}.
\end{remark}

We thus have an upper bound on the entanglement present in $\ket{\varphi^{(BC)}}$ if transformation $\eqref{Eq:trans}$ is going to be accomplished with probability one in finite rounds.  The following theorem shows how much stronger infinite round LOCC is drastically able to surpass this bound.

\begin{theorem}
Let $C(\varphi^{(BC)})=t$.  Transformation \eqref{Eq:trans} is possible in infinite round LOCC with $p_{AB}+p_{AC}+p_{BC}=1$ if and only if $t<1$.
\end{theorem}
\begin{proof}
It has recently been shown that the transformation is deterministically impossible when $t=1$ \cite{Chitambar-2011b}, so to prove the theorem, we will construct an explicit protocol for any $t<1$.  The protocol is analogous to the one given in Ref. \cite{Fortescue-2007a}: each party performs the same measurement one after the other, and either a target state is obtained, or the original state $\ket{W}$ is recovered; in the latter case, the measurement process is repeated.  In fact, the protocol facilitating the desired transformation only requires changing Bob's measurement in the protocol given for Theorem 1.  To recapitulate, Charlie makes the measurement \eqref{Eq:Ops} of $\{M_1(\alpha),M_2(\alpha)\}$ with $\alpha=\frac{2\sqrt{1-t^2}}{1+\sqrt{1-t^2}}>0$.  Outcome 1 is the halt state $\ket{\Phi^{(AB)}}$ and occurs with probability $p_{AB}$.  If the outcome is 2, Alice performs the same measurement $\{M_1(\alpha), M_2(\alpha)\}$.  Outcome 1 is the halt state $\ket{\varphi^{(BC)}}$ and occurs with total probability $p_{BC}$ while outcome 2 is $\ket{\phi_{2,2}}$.  In the new protocol, the third round consists of Bob performing yet again the measurement $\{M_1(\alpha),M_2(\alpha)\}$ on $\ket{\phi_{2,2}}$.  Outcome 1 is the halt state $\ket{\Phi^{(AC)}}$ and occurs with total probability 
\begin{equation}
\label{Eq:PAC}
p_{AC}=\frac{2}{3}\alpha(1-\alpha),
\end{equation}
while outcome 2 is the W state $(\tfrac{1}{3},\tfrac{1}{3},\tfrac{1}{3})$ and occurs with total probability
\begin{equation}
\label{Eq:PW}
p_W=\left(1-\alpha\right)^2.
\end{equation}
In latter case, the entire protocol is then repeated, and this continues for an indefinite number of rounds.  The total probabilities for this infinite round protocol can be calculated from Eqns. \eqref{Eq:PAB}, \eqref{Eq:PBC}, \eqref{Eq:PAC}, and \eqref{Eq:PW}.  For instance we have
\begin{align}
p_{AB}(total)&=\frac{2}{3}\alpha+\left(1-\alpha\right)^2\Bigg(\frac{2}{3}\alpha+ \left(1-\alpha\right)^2\bigg(\frac{2}{3}\alpha+...\notag\\
&=\frac{2}{3}\alpha\sum_{k=0}^\infty\left(1-\alpha\right)^{2k}=\frac{2}{3}\left(\frac{1}{2-\alpha}\right),
\end{align}
and likewise
\begin{align}
p_{BC}(total)&=\left(\frac{2}{3}\alpha-\frac{1}{3}\alpha^2\right)\frac{1}{1-\left(1-\alpha\right)^2}=\frac{1}{3}\notag\\
p_{AC}(total)&=\frac{2}{3}\alpha(1-\alpha)\frac{1}{1-\left(1-\alpha\right)^2}=\frac{2}{3}\left(\frac{1-\alpha}{2-\alpha}\right).
\end{align}
In total we have $p_{BC}(total)+p_{AC}(total)+p_{AB}(total)=1$ which proves the theorem.
\end{proof}

\medskip
\noindent\textbf{Concluding Remarks:}

We have shown that for the range $\sqrt{\frac{1}{2}}< t<1$, transformation \eqref{Eq:trans} requires an infinite number of rounds to be achieved deterministically, while for $0\leq t\leq\sqrt{\frac{1}{2}}$, the transformation can be implemented in only 4 rounds (see Fig. \ref{roundplot}).  Furthermore, if the goal is to distill Ebits across any bipartitioning, say $A:BC$, while still preserving bipartite entanglement between BC after any failure measurement, it is always more useful to perform $n+2$ round of LOCC instead of just $n$.  

A final point to note is that the discontinuity in Fig. \ref{roundplot} only occurs for deterministic transformations ($p_{AB}+p_{AC}+p_{BC}=1$).  If we allow the total success probability to drop to $1-\epsilon$ for any $\epsilon>0$, the Fortescue-Lo Protocol is a finite round protocol that will achieve this transformation, even when $C(\varphi^{(BC)})=1$ \cite{Fortescue-2007a}.  The protocol given in Theorem 2 reduces to this if the parties measure with $\{M_1(\epsilon),M_2(\epsilon)\}$ and Bob and Charlie transform $\ket{\varphi^{(BC)}}\to\ket{\Phi^{(BC)}}$ after Alice's measurement.  However, the potential success of this adds an additional factor of $\tfrac{2-2\epsilon}{2-\epsilon}$ onto $p_{BC}$.  Evaluating the geometric sum up to $n-1$ rounds then yields a total success probability of $\frac{6-4\epsilon}{6-3\epsilon}(1-(1-\epsilon)^{2n})$.  Therefore, the transformation can be achieved with probability at least $1-\epsilon$ by using a protocol having $O(1/\epsilon)$ rounds.

\begin{acknowledgments}
I'd like to thank Andreas Winter for discussing motivating ideas as well as Debbie Leung and Marco Piani for providing helpful comments.  Additional thanks goes to Hoi-Kwong Lo and Wei Cui for carefully reading the article and taking part in related discussions.  This work was supported by funding agencies including CIFAR, the CRC program, NSERC, and QuantumWorks.
\end{acknowledgments}

\bibliography{EricQuantumBib}

\begin{thebibliography}{13}
\expandafter\ifx\csname natexlab\endcsname\relax\def\natexlab#1{#1}\fi
\expandafter\ifx\csname bibnamefont\endcsname\relax
  \def\bibnamefont#1{#1}\fi
\expandafter\ifx\csname bibfnamefont\endcsname\relax
  \def\bibfnamefont#1{#1}\fi
\expandafter\ifx\csname citenamefont\endcsname\relax
  \def\citenamefont#1{#1}\fi
\expandafter\ifx\csname url\endcsname\relax
  \def\url#1{\texttt{#1}}\fi
\expandafter\ifx\csname urlprefix\endcsname\relax\def\urlprefix{URL }\fi
\providecommand{\bibinfo}[2]{#2}
\providecommand{\eprint}[2][]{\url{#2}}

\bibitem[{\citenamefont{Nielsen and Chuang}(2000)}]{Nielsen-2000a}
\bibinfo{author}{\bibfnamefont{M.~A.} \bibnamefont{Nielsen}} \bibnamefont{and}
  \bibinfo{author}{\bibfnamefont{I.~L.} \bibnamefont{Chuang}},
  \emph{\bibinfo{title}{Quantum Computation and Quantum Information}}
  (\bibinfo{publisher}{Cambridge University Press}, \bibinfo{year}{2000}).

\bibitem[{\citenamefont{Plenio and Virmani}(2007)}]{Plenio-2007a}
\bibinfo{author}{\bibfnamefont{M.~B.} \bibnamefont{Plenio}} \bibnamefont{and}
  \bibinfo{author}{\bibfnamefont{S.}~\bibnamefont{Virmani}},
  \bibinfo{journal}{Quant. Inf. Comp.} \textbf{\bibinfo{volume}{7}},
  \bibinfo{pages}{1} (\bibinfo{year}{2007}).

\bibitem[{\citenamefont{Xin and Duan}(2008)}]{Xin-2008a}
\bibinfo{author}{\bibfnamefont{Y.}~\bibnamefont{Xin}} \bibnamefont{and}
  \bibinfo{author}{\bibfnamefont{R.}~\bibnamefont{Duan}},
  \bibinfo{journal}{Phys. Rev. A} \textbf{\bibinfo{volume}{77}},
  \bibinfo{pages}{012315} (\bibinfo{year}{2008}).

\bibitem[{\citenamefont{Bennett et~al.}(1996)\citenamefont{Bennett, DiVincenzo,
  Smolin, and Wootters}}]{Bennett-1996a}
\bibinfo{author}{\bibfnamefont{C.~H.} \bibnamefont{Bennett}},
  \bibinfo{author}{\bibfnamefont{D.~P.} \bibnamefont{DiVincenzo}},
  \bibinfo{author}{\bibfnamefont{J.~A.} \bibnamefont{Smolin}},
  \bibnamefont{and} \bibinfo{author}{\bibfnamefont{W.~K.}
  \bibnamefont{Wootters}}, \bibinfo{journal}{Phys. Rev. A}
  \textbf{\bibinfo{volume}{54}}, \bibinfo{pages}{3824} (\bibinfo{year}{1996}).

\bibitem[{\citenamefont{Lo and Popescu}(2001)}]{Lo-1997a}
\bibinfo{author}{\bibfnamefont{H.-K.} \bibnamefont{Lo}} \bibnamefont{and}
  \bibinfo{author}{\bibfnamefont{S.}~\bibnamefont{Popescu}},
  \bibinfo{journal}{Phys. Rev. A} \textbf{\bibinfo{volume}{63}},
  \bibinfo{pages}{022301} (\bibinfo{year}{2001}).

\bibitem[{\citenamefont{Fortescue and Lo}(2007)}]{Fortescue-2007a}
\bibinfo{author}{\bibfnamefont{B.}~\bibnamefont{Fortescue}} \bibnamefont{and}
  \bibinfo{author}{\bibfnamefont{H.-K.} \bibnamefont{Lo}},
  \bibinfo{journal}{Phys. Rev. Lett.} \textbf{\bibinfo{volume}{98}},
  \bibinfo{pages}{260501} (\bibinfo{year}{2007}).

\bibitem[{\citenamefont{D\"ur et~al.}(2000)\citenamefont{D\"ur, Vidal, and
  Cirac}}]{Dur-2000a}
\bibinfo{author}{\bibfnamefont{W.}~\bibnamefont{D\"ur}},
  \bibinfo{author}{\bibfnamefont{G.}~\bibnamefont{Vidal}}, \bibnamefont{and}
  \bibinfo{author}{\bibfnamefont{J.~I.} \bibnamefont{Cirac}},
  \bibinfo{journal}{Phys. Rev. A} \textbf{\bibinfo{volume}{62}},
  \bibinfo{pages}{062314} (\bibinfo{year}{2000}).

\bibitem[{\citenamefont{Kinta\c{s} and Turgut}(2010)}]{Kintas-2010a}
\bibinfo{author}{\bibfnamefont{S.}~\bibnamefont{Kinta\c{s}}} \bibnamefont{and}
  \bibinfo{author}{\bibfnamefont{S.}~\bibnamefont{Turgut}},
  \bibinfo{journal}{J. Math. Phys.} \textbf{\bibinfo{volume}{51}},
  \bibinfo{pages}{092202} (\bibinfo{year}{2010}).

\bibitem[{\citenamefont{Wootters}(1998)}]{Wootters-1998a}
\bibinfo{author}{\bibfnamefont{W.}~\bibnamefont{Wootters}},
  \bibinfo{journal}{Phys. Rev. Lett.} \textbf{\bibinfo{volume}{10}},
  \bibinfo{pages}{2245} (\bibinfo{year}{1998}).

\bibitem[{\citenamefont{Chitambar et~al.}()\citenamefont{Chitambar, Moriarty,
  and Winter}}]{Chitambar-2011a}
\bibinfo{author}{\bibfnamefont{E.}~\bibnamefont{Chitambar}},
  \bibinfo{author}{\bibfnamefont{T.}~\bibnamefont{Moriarty}}, \bibnamefont{and}
  \bibinfo{author}{\bibfnamefont{A.}~\bibnamefont{Winter}},
  \emph{\bibinfo{title}{Manuscript in preparation}}.

\bibitem[{\citenamefont{Gour et~al.}(2005)\citenamefont{Gour, Meyer, and
  Sanders}}]{Gour-2005a}
\bibinfo{author}{\bibfnamefont{G.}~\bibnamefont{Gour}},
  \bibinfo{author}{\bibfnamefont{D.~A.} \bibnamefont{Meyer}}, \bibnamefont{and}
  \bibinfo{author}{\bibfnamefont{B.~C.} \bibnamefont{Sanders}},
  \bibinfo{journal}{Phys. Rev. A} \textbf{\bibinfo{volume}{72}},
  \bibinfo{pages}{042329} (\bibinfo{year}{2005}).

\bibitem[{\citenamefont{Nielsen}(1999)}]{Nielsen-1999a}
\bibinfo{author}{\bibfnamefont{M.~A.} \bibnamefont{Nielsen}},
  \bibinfo{journal}{Phys. Rev. Lett.} \textbf{\bibinfo{volume}{83}},
  \bibinfo{pages}{436} (\bibinfo{year}{1999}).

\bibitem[{\citenamefont{Chitambar et~al.}(2011)\citenamefont{Chitambar, Cui,
  and Lo}}]{Chitambar-2011b}
\bibinfo{author}{\bibfnamefont{E.}~\bibnamefont{Chitambar}},
  \bibinfo{author}{\bibfnamefont{W.}~\bibnamefont{Cui}}, \bibnamefont{and}
  \bibinfo{author}{\bibfnamefont{H.-K.} \bibnamefont{Lo}}
  (\bibinfo{year}{2011}).

\end{thebibliography}

\end{document}